\documentclass[]{revtex4-2}
\pdfoutput=1
\usepackage{graphicx, amsmath, amssymb, amsfonts, amsthm}
\usepackage{natbib}
\usepackage{listings}
\usepackage[resetlabels]{multibib}
\usepackage{hyperref}
\usepackage{nameref}
\newtheorem{theorem}{Theorem}
\newtheorem{lemma}{Lemma}
\newtheorem{corollary}{Corollary}

\theoremstyle{definition}
\newtheorem{definition}{Definition}

\newtheorem{observation}{Observation}

\newcommand{\outpro}[2]{\vert #1\rangle\langle #2\vert}

\newcommand{\ket}[1]{\vert #1\rangle}

\newcommand{\modulus}[1]{\vert #1 \vert}
\newcommand{\tr}[1]{\mathsf{Tr}(#1)}
\newcommand{\ptr}[2]{\mathsf{Tr}_{#1}(#2)}

\begin{document}


\title{Schmidt Decomposition of Multipartite States}%

\author{Mithilesh Kumar}
 \email{Contact author: mithilesh.qit@gmail.com, mithilesh.kumar@krea.edu.in}
 \homepage{https://drmithileshkumar.com}
\affiliation{
 Krea University, Sri City, India
}



\begin{abstract}
Quantum states can be written in infinitely many ways depending on the choices of basis. The Schmidt decomposition of a quantum state has many useful properties in the study of entanglement. All bipartite states admit Schmidt decomposition, but this does not extend to multipartite systems. We obtain necessary and sufficient conditions for the existence of Schmidt decompositions of multipartite states. Moreover, we provide an efficient algorithm to obtain the decomposition for a Schmidt decomposable multipartite state.
\end{abstract}

\maketitle


\section{\label{sec:intro}Introduction}
Consider Alice and Bob who share a system in some state \(\ket{\psi}\in\mathcal{H}_A\otimes\mathcal{H}_B\). Alice and Bob may choose orthonormal bases \(\{\ket{i_A}\}\) and \(\{\ket{j_B}\}\) independently and express the state as \begin{align}
    \ket{\psi} = \sum_{i,j} a_{ij}\ket{i_A}\ket{j_B}
\end{align} 
such that \(a_{ij}\in \mathbb{C}\) and \(\sum_{ij}|a_{ij}|^2 = 1\). In general, such an expansion may have at most \(n_A\times n_B\) terms. The actual number of terms depends on the chosen bases. Bases that would require minimum number of terms can be obtained by the so-called Schmidt decomposition. Schmidt decomposition provides orthonormal set of vectors (which can be extended to bases) \(\{\ket{k_A}\}\) and \(\{\ket{k_B}\}\) such that 
\begin{align}
    \ket{\psi} = \sum_{k}\lambda_k \ket{k_A}\ket{k_B}
\end{align}
where \(\lambda_k\geq 0\) are real and \(k\leq \min(n_A, n_B)\).

Schmidt decomposition of bipartite states is an immediate consequence of singular value decomposition (SVD) introduced by Schmidt \cite{Schmidt1907}: the matrix \(A\) of coefficients \(a_{ij}\) has SVD as \(A = UDV\). The unitary matrix \(U\) transforms states \(\ket{i_A}\) to \(\ket{k_A}\) and the unitary matrix \(V\) transforms states \(\ket{j_B}\) to \(\ket{k_B}\). For an introduction to Schmidt decomposition of bipartite states and its properties, readers are referred to Nielsen and Chuang \cite{nielsen00}. Some of these properties are listed below:
\begin{enumerate}
    \item \textbf{Minimality}: Schmidt vectors ensure minimum number of terms in the expansion.
    \item \textbf{Bijection}: There is a bijection between the sets \(\{\ket{k_A}\}\) and \(\{\ket{k_B}\}\).
    \item \textbf{Schmidt number}: If the number of terms (known as Schmidt number) is more than one, then the system is entangled.
    \item \textbf{Identical spectrum}: The reduced density matrices \(\rho_A\) and \(\rho_B\) have the same spectrum given by \(\{\lambda_k^2\}\).
    \item \textbf{Local unitary independence}: Schmidt coefficients \(\lambda_k\) are independent of local unitary transformations.
\end{enumerate}
Given Schmidt form of a state, computing entropy or deciding state transformation through majorization or approximation (through keeping the largest $k$ terms) becomes easy. Naturally, one would like to generalize such a decomposition to multipartite states. Although Schmidt decomposition exists for every bipartite state, this does not extend to multipartite states. There are tri-partite states that do not admit such a decomposition. Asher Peres \cite{PERES199516} studied tripartite Schmidt decomposition by obtaining conditions on intermediate matrices after first considering the Schmidt decomposition of a bipartition of the state. Thapliyal \cite{Thapliyal} obtained a condition for Schmidt decomposability of a state and connected it to the separability of the state. Pati \cite{pati2000existence} obtained another criterion for tripartite case using partial inner product. Acín et al. \cite{Ac_n_2000} extended the Schmidt decomposition to three-qubit states. 

Carteret et al. \cite{Carteret} worked on a specific generalization of Schmidt decomposition of multipartite states. In their work, it is shown that a multipartite state can be expressed in orthonormal sets of vectors such that all terms have only real coefficients and local unitary independence is ensured. But rest of properties listed above are not guaranteed. 

In this paper, we want to work with Schmidt decompositions that satisfy all the properties 1-5 listed above. A state of a multipartite system with subsystems labeled by \(A_1, A_2, ..., A_n\) can be written as \begin{align}
    \ket{\psi} = \sum_{i_1i_2...i_n} a_{i_1i_2...i_n}\ket{i_1^{A_1}}\ket{i_2^{A_2}}\cdots\ket{i_n^{A_n}}
\end{align}
A natural generalization of \emph{Schmidt form} of the state is 
\begin{align} \label{eq:multischmidt}
    \ket{\psi} = \sum_\ell \lambda_\ell \ket{\ell_{A_1}} \ket{\ell_{A_2}}\cdots \ket{\ell_{A_n}}
\end{align}
where \(\lambda_\ell\geq 0\) are real, states \(\ket{\ell_{A_1}}, \ket{\ell_{A_2}},\cdots, \ket{\ell_{A_n}}\) are orthonormal in their respective Hilbert spaces. For example, a state of a tripartite system with subsystems \(A, B\) and \(C\) can be written as \begin{align}
    \ket{\psi} = \sum_{ijk} a_{ijk} \ket{i_A}\ket{j_B}\ket{k_C}
\end{align}
with Schmidt form expressed as
\begin{align}\label{eq:trischmidt}
    \ket{\psi} = \sum_\ell \lambda_\ell \ket{\ell_A}\ket{\ell_B}\ket{\ell_C}
\end{align}
where \(\lambda_\ell\geq 0\) are real, states \(\ket{\ell_A}, \ket{\ell_B}\) and \(\ket{\ell_C}\) are orthonormal in their respective Hilbert spaces. Similarly, a state of a quadripartite system with subsystems \(A, B, C\) and \(D\) can be written as \begin{align}
    \ket{\psi} = \sum_{lmno} a_{lmno} \ket{l_A}\ket{m_B}\ket{n_C}\ket{o_D}
\end{align}
with \emph{Schmidt form} of the state as
\begin{align}\label{eq:quadschmidt}
    \ket{\psi} = \sum_k \lambda_k \ket{k_A}\ket{k_B}\ket{k_C}\ket{k_D}
\end{align}
where \(\lambda_k\geq 0\) are real, states \(\ket{k_A}, \ket{k_B}, \ket{k_C}\) and \(\ket{k_D}\) are orthonormal in their respective Hilbert spaces. 

\subsection{Contributions}
 The main contribution of this paper is a necessary and sufficient condition for Schmidt decomposition (Equation \ref{eq:multischmidt}) of multipartite states. In addition, separate such conditions are provided in the special case of tripartite (Equation \ref{eq:trischmidt}) and quadripartite states (Equation \ref{eq:quadschmidt}). Using these conditions, we have provided algorithms to check whether a multipartite state is Schmidt decomposable and provide such a decomposition if it is. Finally, we prove that SCHMIDT-PARTITION problem is NP-complete. 
\section{\label{sec:results}Schmidt decomposition of tripartite states}  
We start with obtaining necessary and sufficient condition for Schmidt decomposition of Equation \ref{eq:trischmidt}. 
For a given matrix \(A\), its adjoint \(A^\dagger\) is a matrix obtained by first taking transpose of \(A\) and then complex conjugating its elements, i.e \begin{align}
    A^\dagger_{ij} = A^*_{ji}
\end{align} 
\begin{definition}
    Two matrices \(A\) and \(B\) are said to \emph{commute} if \(AB = BA\).
\end{definition}
\begin{definition}
    A matrix \(A\) is called \emph{normal} if it commutes with its adjoint, i.e. \(AA^\dagger = A^\dagger A\).
\end{definition}
Normal operators satisfy the following well-known spectral decomposition theorem.
\begin{theorem}
    A matrix is normal if and only if it is diagonalizable, i.e. there exists a unitary matrix \(U\) such that \begin{align}
        A = U^\dagger D U
    \end{align}
    where \(D\) is a diagonal matrix. We say that \(U\) diagonalizes \(A\).
\end{theorem}
\begin{definition}
    A matrix is called positive semi-definite if all of its eigenvalues are non-negative.
\end{definition}
Positive semi-definite matrices are normal and hence spectral decomposition applies to them. The following is a well-known result that 
\begin{theorem}\label{thm:positive}
For any matrix \(A\), the matrices \(A^\dagger A\) and \(AA^\dagger\) are positive semi-definite and have the same eigenvalue spectrum, i.e. there exist unitary matrices \(P\) and \(Q\) such that
\begin{align}
    A^\dagger A &= Q^\dagger D Q\\
    AA^\dagger  &= P D P^\dagger
\end{align}
where \(D\) is diagonal matrix of eigenvalues.
\end{theorem}
Singular value decomposition can be stated as follows
\begin{theorem}\cite{nielsen00}
    For any matrix \(A\), let \(P\) diagonalizes \(AA^\dagger\) and \(Q\) diagonalizes \(A^\dagger A\), then \(P^\dagger A Q^\dagger\) is diagonal.
\end{theorem}
Usually, the spectral decomposition of two matrices \(A\) and \(B\) require different unitary matrices \(U\) and \(V\). When same unitary matrix diagonalizes matrices \(A\) and \(B\), we say \(A\) and \(B\) are \emph{simultaneously} diagonalizable.
\begin{theorem}\label{thm:commute}
    Two normal matrices are simultaneously diagonalizable if and only if they commute.
\end{theorem}
We can extend the notion of commutation to a set of matrices.
\begin{definition}
    A set of matrices \(\mathcal{A}\) is said to commute if every pair of its elements commute.
\end{definition}

\begin{definition}
    A set \(\mathcal{A}\) of \(m\times n\) matrices are said to \emph{positively commute} if for each \(A_i\in \mathcal{A}\)
    \begin{enumerate}
        \item  \(A^\dagger_iA_i\) commute with each other
        \item \(A_iA^\dagger_i\) commute with each other
    \end{enumerate}
\end{definition}
The following lemma applies diagonalization to positively commuting set \(\mathcal{A}\) of matrices.
\begin{lemma}
    If \(\mathcal{A}\) commutes positively, then there exists unitary matrices \(P\) and \(Q\) such that for each \(A_i\in \mathcal{A}\)
    \begin{enumerate}
        \item \(A_i^\dagger A_i = Q^\dagger D_i Q\) and 
        \item \(A_iA_i^\dagger = P D_i P^\dagger\)
    \end{enumerate}
    where \(D_i\) are positive semi-definite and diagonal.
\end{lemma}
\begin{proof}
    Since \(A_i^\dagger A_i\) is positive semi-definite, it admits spectral decomposition given by \begin{align*}
        A_i^\dagger A_i = Q^\dagger D_i Q
    \end{align*}
    where \(D_i\) is a diagonal matrix of non-negative eigenvalues.
    The same unitary matrices work for every \(A_j^\dagger A_j\) as they commute and we can apply Theorem \ref{thm:commute}. The above arguments hold for \(AA^\dagger\) analogously.
\end{proof}
\begin{definition}
    For a positively commuting \(\mathcal{A}\), we call the pair of unitary matrices \((P,Q)\) in the above lemma as diagonalizing pair.
\end{definition}
\begin{definition}
    A matrix \(S\) is called a \emph{scaled unitary} if \(S\) can be decomposed as \(S = \Lambda U\) where
    \begin{enumerate}
        \item \(U\) is unitary,
        \item \(\Lambda\) is positive semi-definite, diagonal and \(\tr{\Lambda^2} = 1\)
    \end{enumerate}
\end{definition}
From the definition, the following lemma follows immediately.
\begin{lemma}
    If \(S\) is scaled unitary, then  \(SS^\dagger\) is diagonal.
\end{lemma}
\begin{definition}
    Given a tripartite state \(\ket{\psi} = \sum_{ijk} a_{ijk}\ket{i_A}\ket{j_B}\ket{k_C}\), we define \emph{matrix set} \(\mathcal{A}\) of \(\ket{\psi}\) as the collection of matrices \(A_i = (a_i)_{jk}\), i.e. we fix index \(i\) and vary indices \(j, k\) to obtain the elements of \(A_i\).
\end{definition}
\begin{definition}
    For any square matrix $A$, $diag(A)$ represents a column with entries from the diagonals of $A$, i.e. $diag(A)_i = A_{ii}$. 
\end{definition}
\begin{theorem}\label{thm:tripartite}
    A tripartite state \(\ket{\psi}\) is Schmidt decomposable if and only if
    \begin{enumerate}
        \item the matrix set \(\mathcal{A}\) commutes positively, and
        \item the matrix \(S = [diag(P^\dagger A_i Q^\dagger)]\) is scaled unitary.
    \end{enumerate}
\end{theorem}
\begin{proof}
    Suppose \(\ket{\psi}\) is Schmidt decomposable. We'll start with the Schmidt decomposition and show that the matrix set \(\mathcal{A}\) commutes positively and \(S\) is scaled unitary. We'll use the unitary relations \(V, P, Q\) between Schmidt bases \(\ket{\ell}\) and starting basis \(\ket{i}\), where \(\Lambda\) is diagonal matrix of coefficients \(\lambda\).
    \begin{align*}
        \ket{\psi} &= \sum_\ell \lambda_\ell \ket{\ell_A} \ket{\ell_B} \ket{\ell_C}\\
                   &= \sum_\ell \left (\Lambda V\ket{i^\ell_A}\right )\ket{\ell_B} \ket{\ell_C}\\
                   &= \sum_\ell \left (\sum_i d^i_{\ell\ell}\ket{i_A}\right )\ket{\ell_B} \ket{\ell_C}\\
                   &= \sum_i \ket{i_A} \sum_\ell d^i_{\ell\ell}\left (\sum_j P_{j\ell}\ket{j_B}\right ) \left (\sum_k Q_{k\ell}\ket{k_C}\right )\\
                   &= \sum_{ijk} \left (\sum_\ell P_{j\ell}d^i_{\ell\ell}Q_{\ell k}\right )\ket{i_A}\ket{j_B}\ket{k_C}\\
                   &= \sum_{ijk}a_{ijk}\ket{i_A}\ket{j_B}\ket{k_C}
    \end{align*}
    The matrices \(A_i\) in \(\mathcal{A}\) are defined as 
    \begin{align}
        A_i = PD_iQ
    \end{align}
    which gives 
    \begin{align}
        A^\dagger_i A_i &= Q^\dagger \modulus{D_i}^2 Q\\
        A_iA^\dagger_i  &= P\modulus{D_i}^2 P^\dagger
    \end{align}
    This implies that each of \(A_i^\dagger A_i\) are diagonalized by \(Q\) and each of \(A_iA_i^\dagger\) are diagonalized by \(P\). By Theorem \ref{thm:commute}, they must commute with each other. Hence, \(\mathcal{A}\) commutes positively. In addition, we define 
    \begin{align}
        S &= \Lambda V\\
          &= [diag(P^\dagger A_i Q^\dagger)]
    \end{align}
    which is scaled unitary. Notice that the normalization condition implies that \(\tr{\Lambda^2} = 1\).

    For the other direction with \(\mathcal{A}\) commuting positively and \(S\) being a scaled unitary, we'll construct the Schmidt decomposition of \(\ket{\psi}\).
    \begin{align*}
        \ket{\psi} &= \sum_{ijk}a_{ijk}\ket{i_A}\ket{j_B}\ket{k_C}\\
                   &= \sum_i \ket{i_A} \sum_{jk} a_{ijk}\ket{j_B}\ket{k_C}\\
                   &= \sum_i \ket{i_A} \sum_{jk} \sum_\ell P_{j\ell} d^i_{\ell\ell} Q_{k\ell} \ket{j_B} \ket{k_C}\\
                   &= \sum_i \ket{i_A} \sum_\ell d^i_{\ell\ell}\left (\sum_{j} P_{j\ell}\ket{j_B}\right )\left (\sum_{k} Q_{k\ell}\ket{k_C}\right )\\
                   &= \sum_\ell \left ( \sum_i d^i_{\ell\ell}\ket{i_A}\right )\ket{\ell_B}\ket{\ell_C}\\
                   &= \sum_\ell \lambda_\ell \ket{\ell_A}\ket{\ell_B}\ket{\ell_C}
    \end{align*}
    Notice that the elements \(d^i_{\ell\ell}\) form the diagonal of decomposition of \(A_i = PD_iQ\). Therefore, the matrix columns of matrix \(S\) formed by elements \(d^i_{\ell\ell}\) are formed by taking diagonal of \(P^\dagger A_i Q^\dagger\) as vector. Since \(S\) is given to be scaled unitary, it has decomposition of the form \(\Lambda V\) where \(\Lambda\) is diagonal and \(V\) is unitary. This provides us the required Schmidt decomposition. 
\end{proof}
\subsection{Algorithm}
Theorem \ref{thm:tripartite} can be used to obtain Schmidt decomposition of a decomposable state efficiently. As the matrix set is diagonalized by matrices \((P, Q)\), we only need to look for an appropriate \(A_i\) to obtain them. Using \(P, Q\), we obtain \(S\). Here are the steps:
\begin{enumerate}
    \item Consider matrices \(A_i\), and compute \(L_i = A_iA_i^\dagger\) and \(M_i = A_i^\dagger A_i\).
    \item Choose random real numbers \(r_i\) and compute \(L = \sum_i r_i L_i\) and \(M = \sum_i r_i M_i\). This random linear combination is to remove degeneracies. 
    \item Compute spectral decompositions of \(L\) and \(M\) to obtain diagonalizing pairs of unitary matrices \(P, Q\).
    \item Compute the diagonal matrices \(D_i = P^\dagger A_i Q^\dagger\). If \(D_i\) are not diagonal, terminate.
    \item Compute the matrix \(S = [diag(D_i)]\).
    \item Compute \(SS^\dagger\), the square root of the diagonal elements gives \(\lambda_\ell\).
    \item Divide the rows of \(S\) by the non-zero \(\lambda_\ell\) which gives the unitary matrix \(V\). If \(V\) is not unitary, terminate.
\end{enumerate}
The correctness of the above algorithm follows from the correctness of Theorem \ref{thm:tripartite}. Since the running time of each step in the algorithm is polynomial in the size of the input, the above algorithm terminates in polynomial time.

\section{Schmidt decomposition of quadripartite states}
The above scheme for tripartite states can be extended to quadripartite states with minor changes as will be shown in this section. To start with, the matrix set \(\mathcal{A}\) needs to account for an extra index.
\begin{definition}
    A set \(\mathcal{D}\) of square matrices is called \emph{unit decomposable} if each matrix \(D_k\in \mathcal{D}\) satisfies following conditions
    \begin{enumerate}
        \item rank of \(D_k\) is 1,
        \item \(D_k\) can be decomposed as \(D_k = \lambda_k u_k v_k^T\) with column vectors $u_k,v_k$ such that \(\lambda_k \geq 0\), \(U = [u_k]\) and \(V = [v_k]\) are unitary.
    \end{enumerate}
\end{definition}
We extend the notion of matrix set defined above to quadripartite states as follows.
\begin{definition}
    Given a quadripartite state \(\ket{\psi} = \sum_{lmno} a_{lmno}\ket{l_A}\ket{m_B}\ket{n_C}\ket{o_D}\), we define \emph{matrix set} \(\mathcal{A}\) of \(\ket{\psi}\) as the collection of matrices \(A^{lm} = (a^{lm})_{no}\), i.e. fix two indices \(l,m\) and vary other indices \(n,o\).
\end{definition}
\begin{theorem}\label{thm:quadripartite}
    A quadripartite state \(\ket{\psi}\) is Schmidt decomposable if and only if
    \begin{enumerate}
        \item the matrix set \(\mathcal{A}\) commutes positively, and
        \item the matrix set \(\mathcal{D}\) with matrices \(D_k = [diag(P^\dagger A^{lm} Q^\dagger)_k]\) is unit decomposable, where \(D_k\) is a matrix formed by taking \(k\)th element of each of vectors \(diag(P^\dagger A^{lm} Q^\dagger)\).
    \end{enumerate}
\end{theorem}
\begin{proof}
    Suppose \(\ket{\psi}\) is Schmidt decomposable.
    \begin{align*}
        \ket{\psi} &= \sum_k \lambda_k \ket{k_A} \ket{k_B} \ket{k_C}\ket{k_D}\\
                   &= \sum_k \lambda_k \left ( \sum_l u^A_{lk}\ket{l_A} \right )  \left( \sum_m u^B_{mk}\ket{m_B}  \right) \ket{k_C} \ket{k_D}\\
                   &= \sum_k \left ( \sum_{l,m} \lambda_k u^A_{lk}u^B_{mk}\ket{l_A} \ket{m_B}  \right) \ket{k_C} \ket{k_D}\\
                   &= \sum_k \left ( \sum_{l,m} d^{lm}_{kk}\ket{l_A} \ket{m_B}  \right) \ket{k_C} \ket{k_D}\\
                   &= \sum_{l,m} \ket{l_A} \ket{m_B} \sum_k d^{lm}_{kk} \ket{k_C} \ket{k_D}\\
                   &= \sum_{l,m} \ket{l_A} \ket{m_B} \sum_k \sum_{n, o} P_{nk} d^{lm}_{kk} Q_{ko} \ket{n_C} \ket{o_D}\\
                   &= \sum_{l,m} \ket{l_A} \ket{m_B} \sum_{n, o} a_{lmno} \ket{n_C} \ket{o_D}
    \end{align*}
    Hence, we have (1) the matrix set \(\mathcal{A}\) commutes positively, and (2) the matrix set \(\mathcal{D}\) with matrices \(D_k = [diag(P^\dagger A^{lm} Q^\dagger)_k]\) is unit decomposable, where \(D_k\) is a matrix formed by taking \(k\)th element of each of vectors \(diag(P^\dagger A^{lm} Q^\dagger)\).
    
    For the other direction we assume that \(\mathcal{A}\) commutes positively and \(D_k\) are unit decomposable. We construct the Schmidt decomposition of the state \(\ket{\psi}\). The proof effectively runs backwards as above.
    \begin{align*}
        \ket{\psi} &= \sum_{l,m,n, o} a_{lmno} \ket{l_A} \ket{m_B}  \ket{n_C} \ket{o_D}\\
        &= \sum_{l,m} \ket{l_A} \ket{m_B} \sum_{n, o} a_{lmno} \ket{n_C} \ket{o_D}\\
        &= \sum_{l,m} \ket{l_A} \ket{m_B} \sum_k \sum_{n, o} P_{nk} d^{lm}_{kk} Q_{ko} \ket{n_C} \ket{o_D}\\ 
        &= \sum_{l,m} \ket{l_A} \ket{m_B} \sum_k d^{lm}_{kk} \ket{k_C} \ket{k_D}\\
        &= \sum_k \left ( \sum_{l,m} d^{lm}_{kk}\ket{l_A} \ket{m_B}  \right) \ket{k_C} \ket{k_D}\\
        &= \sum_k \left ( \sum_{l,m} \lambda_k u^A_{lk}u^B_{mk}\ket{l_A} \ket{m_B}  \right) \ket{k_C} \ket{k_D}\\
        &= \sum_k \lambda_k \left ( \sum_l u^A_{lk}\ket{l_A} \right )  \left( \sum_m u^B_{mk}\ket{m_B}  \right) \ket{k_C} \ket{k_D}\\
        &= \sum_k \lambda_k \ket{k_A} \ket{k_B} \ket{k_C}\ket{k_D}
    \end{align*}
    Decomposition of \(a_{lmno}\) has been according to \(\mathcal{A}\). The matrices defined by \(d^{lm}_{kk}\) considering the diagonals of \(P^\dagger A^{lm} Q^\dagger\) as vectors and used as columns, i.e \(D_k = [diag(P^\dagger A^{lm}Q^\dagger)_k]\).
\end{proof}
\subsection{Algorithm}
Theorem \ref{thm:quadripartite} can be used to obtain Schmidt decomposition of a decomposable state efficiently. As the matrix set is diagonalized by matrices \((P, Q)\), we only need to look for an appropriate \(A^{lm}\) to obtain them. Using \(P, Q\), we obtain matrices \(D_k\). Here are the steps:
\begin{enumerate}
    \item Consider any matrix \(A^{lm}\), possibly with least number of zeroes and compute \(L^{lm} = A^{lm}A^{lm\dagger}\) and \(M^{lm} = A^{lm\dagger} A^{lm}\).
    \item Compute spectral decompositions of random linear combinations of \(L^{lm}\) and \(M^{lm}\) to obtain diagonalizing pairs of unitary matrices \(P, Q\).
    \item Compute the diagonal matrices \(S^{lm} = P^\dagger A^{lm} Q^\dagger\). If any of them is not diagonal, terminate.
    \item Compute the matrix \(D_k = [diag(S^{lm})_k]\).
    \item Compute rank 1 decompositions of \(D_k = \lambda_k u^A_k u^{BT}_k\). If any one of $D_k$ are not rank 1, terminate.
    \item Output \(\lambda_k, U^A = [u^A_k], U^B = [u^B_k], P\) and \(Q\), that can be used to construct the Schmidt decomposition.
\end{enumerate}
The correctness of the above algorithm follows from the correctness of Theorem \ref{thm:quadripartite}. Since the running time of each step in the algorithm is polynomial in the size of the input, the above algorithm terminates in polynomial time.

\section{Schmidt decomposition of multipartite states}
In this section, we generalize the schemes for tripartite and quadripartite states to general multipartite states.
\begin{definition}
    A family \(\mathcal{M}\) of \(N\) positively commuting sets \(\mathcal{A}_i\) of matrices is called \emph{central} if the diagonalizing pairs \((P_i, Q_i)\) for each set \(\mathcal{A}_i\) is of the form \((P_i, Q_{N})\).
\end{definition}
\begin{definition}
    For given set of indices \(\{i_1, i_2, ..., i_n\}\), define \(S_{i_ji_k}\) as set of indices \(\{i_1, i_2, ..., i_n\} - \{i_j, i_k\}\).
    Given a set of numbers \(a_{i_1i_2...i_n}\), define a matrix \(A^{S_{i_ji_k}}\)  by fixing the indices \(S_{i_ji_k}\) and only varying \(i_j\) and \(i_k\). 
\end{definition}
\begin{definition}
    Given a multipartite state \(\ket{\psi} = \sum_{i_1i_2...i_n} a_{i_1i_2...i_n}\ket{i_1^{A_1}}\ket{i_2^{A_2}}\cdots\ket{i_n^{A_n}}\), we define \emph{matrix set} of \(\ket{\psi}\) as the family \(\mathcal{M}\) of collection \(\mathcal{A}_\ell\) of matrices \(A^{S_{i_\ell, i_n}}\).
\end{definition}
\begin{theorem}\label{thm:main}
    A multipartite state \(\ket{\psi}\) is Schmidt decomposable if and only if its matrix set \(\mathcal{M}\) is central.
\end{theorem}
\begin{proof}
    Suppose \(\ket{\psi}\) is Schmidt decomposable.
    \begin{align*}
        \ket{\psi} &= \sum_\ell \lambda_\ell \ket{\ell_{A_1}} \ket{\ell_{A_2}}\cdots \ket{\ell_{A_n}}\\
                   &= \sum_\ell \lambda_\ell \left ( U_{A_1}\ket{i^{A_1}_\ell}\right )\left ( U_{A_2}\ket{i^{A_2}_\ell}\right )\cdots \left ( U_{A_n}\ket{i^{A_n}_\ell}\right )\\
                   &= \sum_\ell \lambda_\ell \left ( \sum_{i_1}u^{A_1}_{i_1\ell}\ket{i^{A_1}_1}\right )\left ( \sum_{i_2}u^{A_2}_{i_2\ell}\ket{i^{A_2}_2}\right )\cdots\\& \left ( \sum_{i_n}u^{A_n}_{i_n\ell}\ket{i^{A_n}_n}\right )\\
                   &= \sum_{i_1i_2...i_n}\left (\sum_\ell \lambda_\ell u^{A_1}_{i_1\ell}u^{A_2}_{i_2\ell}\cdots u^{A_n}_{i_n\ell}\right ) \ket{i_1^{A_1}}\ket{i_2^{A_2}}\cdots\ket{i_n^{A_n}}\\
    \end{align*}
    This implies that 
    \begin{align}\label{eq:factor}
        a_{i_1i_2...i_n} = \sum_\ell \lambda_\ell u^{A_1}_{i_1\ell}u^{A_2}_{i_2\ell}\cdots u^{A_n}_{i_n\ell}
    \end{align}
    From this decompositions of \(a_{i_1i_2...i_n}\), we can define any matrix \(A^{S_{i_k i_n}}\) with elements 
    \begin{align}
        a^{S_{i_k i_n}}_{i_k i_n} = \sum_\ell u^{A_k}_{i_k\ell} d^{S_{i_k i_n}}_{\ell \ell} u^{A_n}_{\ell i_n}
    \end{align}
    This shows that we can decompose 
    \begin{align}
        A^{S_{i_k i_n}} = U^{A_k}D^{S_{i_k i_n}} U^{A_n}
    \end{align}
    Hence, the sets \(\mathcal{A}_k\) positively commute and diagonalizing pairs are of the form \((U^{A_k}, U^{A_n})\). This implies that the matrix set \(\mathcal{M}\) is central.

    For the other direction, we assume that the matrix set \(\mathcal{M}\) is central.
    \begin{align*}
        \ket{\psi} &= \sum_{i_1i_2...i_n} a_{i_1i_2...i_n}\ket{i_1^{A_1}}\ket{i_2^{A_2}}\cdots\ket{i_n^{A_n}}\\
        &= \sum_{S_{i_k, i_n}}\ket{S_{i_k, i_n}}\sum_{i_k, i_n} a_{i_1i_2...i_n}\ket{i_k}\ket{i_n}\\
        &= \sum_{S_{i_k, i_n}}\ket{S_{i_k, i_n}}\sum_{i_k, i_n} \left (\sum_\ell u^{A_k}_{i_k\ell}d^{S_{i_k, i_n}}_{\ell \ell} u^{A_n}_{\ell i_n}\right )\ket{i_k}\ket{i_n}\\
        &= \sum_{S_{i_k, i_n}}\ket{S_{i_k, i_n}}\sum_{i_k} \left (\sum_\ell u^{A_k}_{i_k\ell}d^{S_{i_k, i_n}}_{\ell \ell} \right )\ket{i_k}\left (\sum_{i_n} u^{A_n}_{\ell i_n}\ket{i^{A_n}_n}\right )\\
        &= \sum_\ell\left ( \sum_{S_{i_n}} u^{A_k}_{i_k\ell}d^{S_{i_k, i_n}}_{\ell \ell} \ket{S_{i_n}}\right ) \ket{\ell_{A_n}}\\
    \end{align*}
    The above decomposition can be obtained for any \(1\leq k\leq n-1\). This implies that
    \begin{align}
         u^{A_1}_{i_1\ell}d^{S_{i_1, i_n}}_{\ell \ell} = u^{A_2}_{i_1\ell}d^{S_{i_2, i_n}}_{\ell \ell} = \cdots = u^{A_{n-1}}_{i_1\ell}d^{S_{i_{n-1}, i_n}}_{\ell \ell} = z^{S_{i_n}}_\ell
    \end{align}
    This implies that each of the indices in \(z^{S_{i_n}}_\ell\) can be factored. This must be applicable for each \(d^{S_{i_k, i_n}}_{\ell \ell}\). Therefore, we can write
    \begin{align}
        z^{S_{i_n}}_\ell &= \lambda_\ell u^{A_1}_{i_1\ell}u^{A_2}_{i_2\ell}\cdots u^{A_{n-1}}_{i_{n-1}\ell}
    \end{align}
    The above factors imply Equation \ref{eq:factor} and hence provide the Schmidt decomposition of \(\ket{\psi}\). This concludes the proof of the theorem.
\end{proof}
\subsection{Algorithm}
Theorem \ref{thm:main} can be used to obtain Schmidt decomposition of a decomposable state efficiently. The matrix set is diagonalized by matrix pairs \((P_i, Q_n)\). The set of unitary matrices \(\{P_1, P_2, ..., P_{n-1}, Q_n\}\) obtained by considering matrices \(A^{S_{i_k, i_n}}\) form the transformation \(P_1\otimes P_2\otimes\cdots Q_n\) of \(\ket{\psi}\) that gives the Schmidt basis. The coefficients \(\lambda_\ell\) can be obtained by solving a system of \(\ell\) linear equations using Equation \ref{eq:factor}. \\
The correctness of the above algorithm follows from the correctness of Theorem \ref{thm:main}. Since the running time of each step in the algorithm is polynomial in the size of the input, the above algorithm terminates in polynomial time.
\section{Classification based on Schmidt bases}\label{sec:schmidtBases}

This section appeared in arXiv pre-print \cite{kumar2024propertiesschmidtdecomposition}. As is noted in \cite{nielsen00}, for a bipartite state, if \(\ket{\psi} = \sum_i\lambda_i\ket{i_A}\ket{i_B}\) is the Schmidt decomposition, then \(\sum_i\lambda_i(U\ket{i_A})(V\ket{i_B})\) are Schmidt decomposition of \((U\otimes V)\ket{\psi}\). This observation extends to multipartite systems as well. Moreover, if two states \(\ket{\psi}\) and \(\ket{\phi}\) have the same Schmidt coefficients, then they must be unitary transforms of each other, i.e. \(\ket{\psi} = (U\otimes V)\ket{\phi}\).
\begin{theorem}\label{thm:multiUnitary}
    Suppose \(\ket{\psi}\) and \(\ket{\phi}\) are two Schmidt decomposable multipartite states on the same Hilbert space. Then, \(\ket{\psi} = (U_1\otimes U_2\otimes \cdots U_n)\ket{\phi}\) if and only if the Schmidt decompositions of \(\ket{\psi}\) and \(\ket{\phi}\) have the same Schmidt coefficients.
\end{theorem}
\begin{proof}
    If \(\ket{\psi} = (U_1\otimes U_2\otimes \cdots U_n)\ket{\phi}\), we can use the Schmidt decomposition of \(\ket{\phi}\) as 
    \begin{align*}
        \ket{\phi} &= \sum_\ell \lambda_\ell\ket{\ell_{A_1}}\ket{\ell_{A_2}}\cdots\ket{\ell_{A_n}}\\
        \ket{\psi} &= \sum_\ell \lambda_\ell(U_{A_1}\ket{\ell_{A_1}})(U_{A_2}\ket{\ell_{A_2}})\cdots(U_{A_n}\ket{\ell_{A_n}})\\
         &= \sum_\ell \lambda_\ell\ket{\ell^\psi_{A_1}}\ket{\ell^\psi_{A_2}}\cdots\ket{\ell^\psi_{A_n}}
    \end{align*}
    As states \(\ket{\ell^\psi_{A_k}}\) are obtained via unitary transform of orthonormal states \(\ket{\ell_{A_k}}\), they remain orthonormal and hence provide a valid Schmidt decomposition of \(\ket{\psi}\). Therefore, both \(\psi\) and \(\phi\) have the same Schmidt coefficients.

    For converse, assume \(\psi\) and \(\phi\) have the same Schmidt coefficients, i.e.
    \begin{align*}
        \ket{\phi} &= \sum_\ell \lambda_\ell\ket{\ell_{A_1}}\ket{\ell_{A_2}}\cdots\ket{\ell_{A_n}}\\
        \ket{\psi} &= \sum_\ell \lambda_\ell\ket{\ell^\psi_{A_1}}\ket{\ell^\psi_{A_2}}\cdots\ket{\ell^\psi_{A_n}}
    \end{align*}
    Since states \(\{\ket{\ell_{A_k}}\}\) and \(\{\ket{\ell^\psi_{A_k}}\}\) are orthonormal vectors, there exists unitary matrix \(U_{A_k}\) such that \(\ket{\ell^\psi_{A_k}} = U_{A_k}\ket{\ell_{A_k}}\). This implies that 
    \begin{align*}
        \ket{\psi} &= \sum_\ell \lambda_\ell\ket{\ell^\psi_{A_1}}\ket{\ell^\psi_{A_2}}\cdots\ket{\ell^\psi_{A_n}}\\
        &=\sum_\ell \lambda_\ell(U_{A_1}\ket{\ell_{A_1}})(U_{A_2}\ket{\ell_{A_2}})\cdots(U_{A_n}\ket{\ell_{A_n}})\\
        &= (U_1\otimes U_2\otimes \cdots U_n) \sum_\ell \lambda_\ell\ket{\ell_{A_1}}\ket{\ell_{A_2}}\cdots\ket{\ell_{A_n}}\\
        &= (U_1\otimes U_2\otimes \cdots U_n)\ket{\phi}
    \end{align*}
\end{proof}
Theorem \ref{thm:multiUnitary} implies that if we define Schmidt decomposable states \(\ket{\psi}\) and \(\ket{\phi}\) \emph{equivalent} if \(\ket{\psi} = (U_1\otimes U_2\otimes \cdots U_n)\ket{\phi}\), then the Schmidt number is decided by the number of separable orthonormal basis states. We can consider Schmidt decomposition as establishing a bijection between orthonormal set of vectors of each component space. 
\begin{figure}
    \centering
    \includegraphics[width=0.3\linewidth]{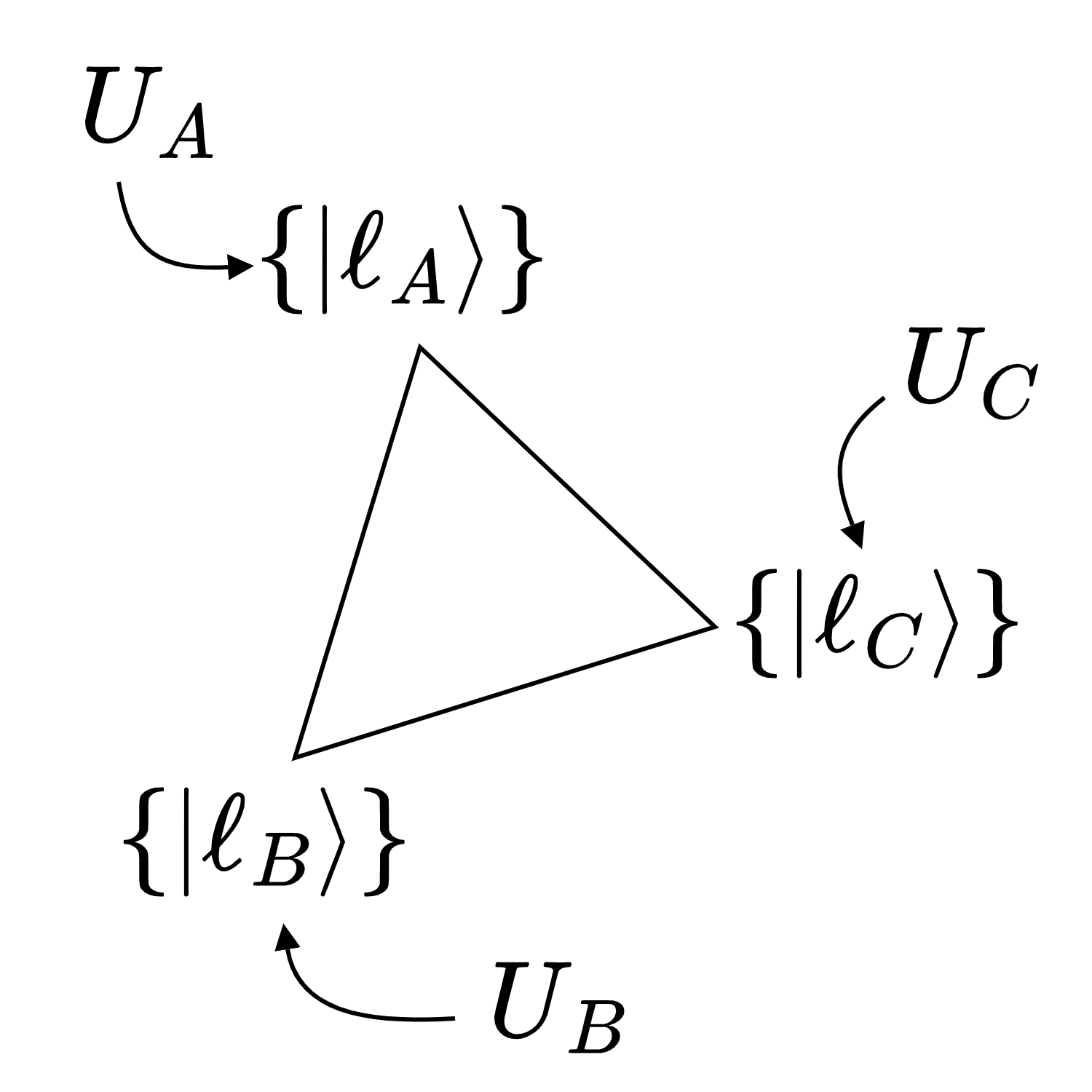}
    \caption{Schmidt decomposition establishes a bijection between orthonormal set of vectors of each subsystem. Equivalent states having the same Schmidt coefficients are generated using unitary operations on individual subsystems.}
    \label{fig:enter-label}
\end{figure}

\begin{observation}
    All two qubit Schmidt bases with Schmidt number \(1\) are equivalent to \(\ket{00}\) and Schmidt number \(2\) are equivalent to \(\{\ket{00},\ket{11}\}\).
\end{observation}
In case of three qubit system, we can partition the system in two ways, (a) one subsystem with \(1\) qubit and another with \(2\) qubits, and (b) three subsystems with \(1\) qubit each. In both the cases, Schmidt rank \(1\) bases are \(\ket{000}\). Schmidt rank \(2\) bases are \(\{\ket{000},\ket{111}\}\). 
\begin{observation}
    There are no Schmidt rank \(3\) states on a three qubit system (irrespective of partition).
\end{observation}
\begin{observation}
    The Schmidt number of an \(n\) qubit system is bounded by \(2^{\lfloor \frac{n}{2}\rfloor}\) irrespective of partition. In fact, the highest Schmidt number is achieved for equal bipartition of the system.
\end{observation}
Let us define the problem SCHMIDT-PARTITION  as follows:\\
\textbf{Input:} \(n\) systems of with \(d_1, d_2,...,d_n\) qubits\\
\textbf{Task:} partition the system into two parts such that there exists a state of Schmidt number \(K\)

\begin{theorem}\label{thm:Spart}
    SCHMIDT-PARTITION is \texttt{NP}-complete.
\end{theorem}
\begin{proof}
    To have a state of the Schmidt number \(K\), the dimension of the smallest subsystem must be at least \(K\). The dimension of combined system of two subsystems of dimensions \(2^{d_1}\) and \(2^{d_2}\) is \(2^{d_1+d_2}\). Since all dimensions are positive numbers, we can take \(\log\) of the dimensions and ask whether there exists a bipartition of the system such that the sum of \(d_i\) for both parts is at least \(\log(K)\) each. This is equivalent to the \texttt{NP}-complete PARTITION problem in which the input is a set of positive numbers and the task is to partition the set into two subsets whose sum is at least \(K\). It is important to note that not all instances of PARTITION are \texttt{NP}-complete. As the dimensions $2^{d_i}$ can be arbitrarily large, even $d_i$ can be large making the corresponding instance hard.
\end{proof}
The Theorem \ref{thm:Spart} immediately implies the following corollary.
\begin{corollary}
    Given \(n\) systems of \(d_1, d_2,...,d_n\) qubits, it is \texttt{NP}-complete to find a Schmidt decomposable state of highest Schmidt number.
\end{corollary}
\begin{theorem}\label{thm:reduced}
    If \(\ket{\psi} = \sum_{ij}a_{ij}\ket{i_A}\ket{j_B}\) is a bipartite state where \(\ket{i_A}\) and \(\ket{j_B}\) for orthonormal bases of their respective spaces, then the reduced density matrix of any subsystem is given by 
    \begin{align}
        \rho^A &= AA^\dagger\\
        \rho^B &= BB^\dagger = (A^\dagger A)^T
    \end{align}
    where \(A\) is the matrix of coefficients \(a_{ij}\) and \(B = A^T\).
\end{theorem}
\begin{proof}
The density matrix of the system is 
    \begin{align*}
        \rho &= \outpro{\psi}{\psi}\\
        &= \sum_{ijkl}a_{ij}a^*_{kl}\outpro{i_Aj_B}{k_Al_B}
    \end{align*}
Tracing out B gives
\begin{align*}
    \rho_A &= \sum_{ijk}a_{ij}a^*_{kj}\outpro{i_A}{k_A}\\
    &= \sum_{ik}\left(\sum_j a_{ij}a^*_{kj}\right )\outpro{i_A}{k_A}\\
    &= AA^\dagger
\end{align*}
Tracing out A gives
\begin{align*}
    \rho_B &= \sum_{ijl}a_{ij}a^*_{il}\outpro{j_B}{l_B}\\
    &= \sum_{jl}\left(\sum_i a^*_{il}a_{ij}\right )\outpro{j_B}{l_B}\\
    &= A^T (A^T)^\dagger = BB^\dagger
\end{align*}
\end{proof}
Theorem \ref{thm:reduced} can be extended to multipartite states as well. We show this in case of tripartite states. Similar pattern can be used for higher partite states.
\begin{theorem}\label{thm:reducedTripartite}
    If \(\ket{\psi} = \sum_{ijk}a_{ijk}\ket{i_A}\ket{j_B}\ket{k_C}\) is a tripartite state where \(\ket{i_A}, \ket{j_B}\) and \(\ket{k_C}\) for orthonormal bases of their respective spaces, then the reduced density matrices of one-level and two-level subsystems are given by 
    \begin{align}
        \rho^A &= AA^\dagger\\
        \rho^{AB} &= BB^\dagger
    \end{align}
    where \(B\) is the matrix of coefficients \(a_{ij,k}\) of dimension \(n_An_B\times n_C\) and \(A\) is the matrix of coefficients \(a_{i, jk}\) of dimension \(n_A\times n_Bn_C\).
\end{theorem}
\begin{proof}
    Let us calculate \(\rho_A\) first by tracing over \(B\) and \(C\) in the matrix
    \begin{align}\label{eq:tridensity}
        \outpro{\psi}{\psi} = \sum_{ijk, lmn} a_{ijk}a^*_{lmn}\outpro{i_Aj_Bk_C}{l_Am_Bn_C}
    \end{align}
    Tracing out subsystems \(B\) and \(C\) gives
    \begin{align*}
        \rho_A &= \sum_{ijkl}a_{ijk}a^*{ljk}\outpro{i_A}{l_A}\\
               &= \sum_{il}\left(\sum_{jk}a_{ijk}a^*_{ljk}\right )\outpro{i_A}{l_A}\\
               &= \sum_{il}\left(\sum_{jk}a_{i,jk}a^\dagger_{jk,l}\right )\outpro{i_A}{l_A}
    \end{align*}
    In the inner sum, we treat the indices \(jk\) as one giving us a matrix \(A\) of dimension \(n_A\times n_Bn_C\). Therefore,
    \begin{align*}
        \rho_A = AA^\dagger
    \end{align*}
    Next, tracing out \(C\) only gives
    \begin{align*}
        \rho_{AB} &= \sum_{ijk,lm}a_{ijk}a^*_{lmk}\outpro{i_Aj_B}{l_Am_B}\\
                  &= \sum_{ij,lm}\left (\sum_k a_{ij,k}a^*_{lm,k}\right )\outpro{i_Aj_B}{l_Am_B}
    \end{align*}
    In the inner sum, we define matrix \(B\) of elements \(a_{ij, k}\) of dimension \(n_An_B\times n_C\). Therefore,
    \begin{align*}
        \rho_{AB} = BB^\dagger
    \end{align*}
    Reduced density matrices of other subsystems can be obtained by relabeling the indices appropriately.
\end{proof}
It is easy to see that the Schmidt number of a bipartite state is equal to the rank of the reduced density matrix of one of the subsystems by using \(rank(A) = rank(AA^\dagger) = rank(A^\dagger A) = rank((A^\dagger A)^T)\). This result holds for multipartite states as well. 
\begin{theorem}
    If \(\ket{\psi}\) is a Schmidt decomposable multipartite state, then the Schmidt number is equal to the rank of the reduced density matrix of one of the subsystems.
\end{theorem}
\begin{proof}
We show here for tripartite states.
    Consider the Schmidt decomposition of \(\ket{\psi} = \sum_\ell\lambda_\ell\ket{\ell_A}\ket{\ell_B}\ket{\ell_C}\).
    \begin{align*}
        \rho^{ABC} &= \outpro{\psi}{\psi}\\
                   &= \sum_{\ell,m}\lambda_\ell\lambda_m\outpro{\ell_A}{m_A}\otimes\outpro{\ell_B}{m_b}\otimes\outpro{\ell_C}{m_C}\\
        \rho^{AB} &= \sum_\ell^d\lambda_\ell^2\outpro{\ell_A\ell_B}{\ell_A\ell_B}
    \end{align*}
    Consider the spectral decomposition of \(\rho^{AB}\) given as 
    \begin{align*}
        \rho^{AB} = \sum_i^r p_i\outpro{i_{AB}}{i_{AB}}
    \end{align*}
    Since the states \(\{\ket{\ell_A\ell_B}\}\) and \(\{\ket{i_{AB}}\}\) are orthonormal sets, they can each be extended to orthonormal bases using Gram-Schmidt orthogonalization process. This implies that there is unitary matrix \(U\) that takes set of vectors \(\{\ket{\ell_A\ell_B}\}\) to \(\{\ket{i_{AB}}\}\).
    \begin{align*}
        U\rho^{AB}U^\dagger = \sum_i^r p_i\outpro{\ell_A\ell_B}{\ell_A\ell_B}
    \end{align*}
    Since \(\rho^{AB}\) and \(U\rho^{AB}U^\dagger\) are similar matrices, they must have the same rank, i.e. \(d = r\).
\end{proof}
In the book \cite{nielsen00} (attributed to Thapliyal), for a bipartite state \(\ket{\psi} = \alpha \ket{\phi} + \beta\ket{\gamma}\), the Schmidt number satisfy the inequality \(Sch(\psi)\geq |Sch(\phi) - Sch(\gamma)|\). Tt is a simple consequence of Theorem \ref{thm:reduced}. We show that it is true for multipartite Schmidt decomposable states. 
\begin{theorem}
    Given three multipartite Schmidt decomposable states \(\ket{\psi}, \ket{\phi}\) and \(\ket{\gamma}\) such that \(\ket{\psi} = \alpha \ket{\phi} + \beta\ket{\gamma}\), then 
    \begin{align}
        Sch(\psi)\geq |Sch(\phi) - Sch(\gamma)|
    \end{align}
\end{theorem}
\begin{proof}
    We show it for tripartite state for readability.
    \begin{align*}
    \ket{\psi} &= \alpha\ket{\phi} + \beta\ket{\gamma}\\ &= \sum_{ijk}(\alpha a^\phi_{ijk} + \beta a^\gamma_{ijk})\ket{i_A}\ket{j_B}\ket{k_C}\\ &= \sum_{ijk}a_{ijk}\ket{i_A}\ket{j_B}\ket{k_C}
\end{align*}
Using Theorem \ref{thm:reducedTripartite}, we get 
\begin{align*}
    \rho_A &= AA^\dagger\\
           &= (\alpha A_\phi + \beta A_{\gamma})(\alpha A_\phi + \beta A_{\gamma})^\dagger
\end{align*}
This implies that
\begin{align*}
    Sch(\psi) &= rank(\rho_A)\\
              &= rank((\alpha A_\phi + \beta A_{\gamma})(\alpha A_\phi + \beta A_{\gamma})^\dagger)\\
              &= rank(\alpha A_\phi + \beta A_{\gamma})\\
              &\geq |rank(\alpha A_\phi) - rank(\beta A_{\gamma})|\\
              &= |rank(A_\phi) - rank(A_{\gamma}) |\\
              &= |Sch(\phi) - Sch(\gamma)|
\end{align*}
\end{proof}

\begin{observation}
    Linear combination of Schmidt decomposable multipartite states is not necessarily Schmidt decomposable. An an example, we can work with \(\ket{W} = \frac{1}{\sqrt{3}}(\ket{001}) + \ket{010} + \ket{100}\) which is not Schmidt decomposable, but each of the terms are Schmidt decomposable.
\end{observation}
In light of the above observation, we prove the following result for tensor product of Schmidt decomposable states.
\begin{definition}
    An \(n\)-partite Schmidt decomposable state of Schmidt rank \(k\) is said to have Schmidt dimension \((n,k)\).
\end{definition}
\begin{theorem}
    Given two states \(\ket{\psi}\) and \(\ket{\phi}\) of Schmidt dimensions \((m,k_\psi)\) and \((n, k_\phi)\) such that \(m\geq n\), then \(\ket{\psi}\ket{\phi}\) has dimension \((n, k_\psi k_\phi)\) and there are \({m-1\choose n-1}\) ways to construct the \(n\)-partite systems.
\end{theorem}
\begin{proof}
    Starting with the Schmidt decompositions of \(\ket{\psi}\) and \(\ket{\phi}\)
    \begin{align*}
        \ket{\psi} &= \sum_{i}^{k_\psi} \lambda_i^\psi \ket{i_{A_1}...i_{A_m}}\\
        \ket{\phi} &= \sum_{j}^{k_\phi} \lambda_j^\phi \ket{j_{B_1}...i_{B_n}}\\
        \ket{\psi}\ket{\phi} &= \sum_{i,j}^{k_\psi, k_\phi} \lambda_i^\psi \lambda_j^\phi \ket{i_{A_1}...i_{A_m}j_{B_1}...j_{B_n}}
    \end{align*}
    Notice that the above expression for \(\ket{\psi}\ket{\phi}\) is not a Schmidt decomposition of \((m + n)\)-partite system as we lose the bijection between states of each part. We partition \(m\)-parts of \(\ket{\psi}\) into \(n\) groups such that each group has at least 1 part. Let \(x_1, ..., x_n\) be the sizes of each groups. For each group, we link one part from \(\ket{\phi}\) as shown below:
    \begin{align*}
        \ket{\psi}\ket{\phi} &= \sum_{i,j}^{k_\psi, k_\phi} \lambda_i^\psi \lambda_j^\phi \ket{i_{A_1}...i_{A_{x_1}}j_{B_1}}\ket{i_{A_{x_1 + 1}}...A_{x_1 + x_2}j_{B_2}}\\&...\ket{i_{m-(x_n) + 1}...i_{m}j_{B_n}}
    \end{align*}
    It is easy to verify that states of each group are orthonormal and they don't repeat in the sum, giving a valid Schmidt decomposition. Hence, \(\ket{\psi}\ket{\phi}\) can be partitioned into \(n\) parts such that the \(n\)-partite system has Schmidt decomposition of rank \(k_\psi k_\phi\).

    The number of such groups for fixed permutation of parts in \(\ket{\psi}\) and \(\ket{\phi}\) is equal to the number of solutions of 
    \begin{align*}
        x_1 + x_2 + \cdots + x_n &= m\\
        x_i &\geq 1
    \end{align*}
    which is given by \({m-1\choose n-1}\).
\end{proof}
\section{Purification}
Given a density matrix \(\rho\) of a system \(A\), the task of purification is to attach another subsystem \(R\) such that the combine system is in pure state \(\ket{AR}\) and the tracing out \(R\) gives back \(\rho\). The standard trick \cite{Schrödinger_1936, HUGHSTON199314, Jaynes, Hadjisavvas, Gisin} is to keep the dimension of \(R\) to be at least the rank of \(\rho\) and an embed orthonormal basis into \(\rho\) as follows starting with the spectral decomposition of \(\rho\).
\begin{align*}
    \rho &= \sum_i \lambda_i \outpro{i_A}{i_A}\\
    \ket{AR} &= \sum_i \sqrt{\lambda_i}\ket{i_A}\ket{i_R}\\
    \rho &= \ptr{R}{\outpro{Ar}{Ar}}\\
         &= \sum_i \lambda_i \outpro{i_A}{i_A}
\end{align*}
Essentially, we have used the Schmidt decomposition of \(\ket{AR}\).

Regarding purification, one can ask some interesting questions like given a density matrix of a composite system, can it be purified to a Schmidt decomposable state? If a multipartite state is not Schmidt decomposable, can an additional system make it Schmidt decomposable?
\begin{lemma}\cite{nielsen00}\label{thm:twopure}
    Two purifications \(\ket{AR_1}\) and \(\ket{AR_2}\) of a state \(\rho\) are unitarily linked as \begin{align*}
        \ket{AR_1} = (I\otimes U_R)\ket{AR_2}
    \end{align*} 
\end{lemma}
\begin{proof}
    The key idea is that two orthonormal sets are unitarily connected. Using the spectral decomposition of \(\rho\), we have \(\rho = \sum_i \lambda_i\outpro{i_A}{i_A}\).
    \begin{align*}
        \ket{AR_1} &= \sum_i\sqrt{\lambda_i}\ket{i_A}\ket{i^R_1}\\
        \ket{AR_2} &= \sum_i\sqrt{\lambda_i}\ket{i_A}\ket{i^R_2}\\
    \end{align*}
    Since the orthonormal sets \(\{\ket{i^R_1}\}\) and \({\ket{i^R_2}}\) can be connected by unitary matrix \(U\) such that \(\ket{i^R_1} = U\ket{i^R_2}\), we have
    \begin{align*}
        \ket{AR_1} &= \sum_i\sqrt{\lambda_i}\ket{i_A}\ket{i^R_1}\\
                   &= \sum_i\sqrt{\lambda_i}\ket{i_A}U\ket{i^R_2}\\
                   &= (I\otimes U) \sum_i\sqrt{\lambda_i}\ket{i_A}\ket{i^R_2}\\
                   &= (I\otimes U)\ket{AR_2}
    \end{align*}
\end{proof}
\begin{theorem}
    Every purification of a multipartite system \(\rho\) is either Schmidt decomposable or not Schmidt decomposable.
\end{theorem}
\begin{proof}
    The proof follows by combining Lemma \ref{thm:twopure} with Theorem \ref{thm:multiUnitary}.
\end{proof}
\section{Conclusion}
In this paper, we obtained a necessary and sufficient condition for the existence of Schmidt decomposition of multipartite states. This condition is constructive and is used to obtain an efficient algorithm for obtaining the decomposition of a Schmidt decomposable state.

\bibliographystyle{plainnat}
\bibliography{main}

\end{document}